\documentclass[authoryear,10pt]{elsarticle}

\usepackage{amssymb}
\usepackage{amsmath}
\usepackage{amsthm}
\newtheorem{thm}{Theorem}

\newtheorem{rem}{Remark}

\begin{document}

\begin{frontmatter}



\title{Using mixtures in seemingly unrelated linear regression models with non-normal errors}


 \author{Giuliano Galimberti\corref{cor1}}
 \ead{giuliano.galimberti@unibo.it}
 \author{Elena Scardovi}
 \author{Gabriele Soffritti}

\cortext[cor1]{Correspondence to: Department of Statistical Sciences, University of Bologna \\ via Belle Arti 41, 40126 Bologna, Italy. Tel.: +39 051 2098227, Fax: +39 051 232153}


\address{Department of Statistical Sciences, University of Bologna}

\begin{abstract}
Seemingly unrelated linear regression models are introduced in which the distribution of the errors is a finite mixture of Gaussian components. Identifiability conditions are provided. The score vector and the Hessian matrix are derived. Parameter estimation is performed using the maximum likelihood method and an Expectation-Maximisation algorithm is developed. The usefulness of the proposed methods and a numerical evaluation of their properties are illustrated through the analysis of a real dataset.
\end{abstract}


\end{frontmatter}


\section{Introduction}\label{intro}
``Seemingly unrelated regression equations'' is an expression first used by \citet{zellner1962}. It indicates a set of equations for modelling the dependence of $D$ variables ($D \geq 1$) on one or more regressors in which the error terms in the different equations are allowed to be correlated and, thus, the equations should be jointly considered.
The range of situations for which models composed of seemingly unrelated regression equations are appropriate is wide, including cross-section data, time-series data and repeated measures \citep[see, e.g.,][]{srivastava1987,park1993}.

Seemingly unrelated regression models have been studied through many approaches.
In \citet{zellner1962,zellner1963} feasible generalized least squares estimators are introduced and their properties are analysed. The maximum likelihood estimator from a Gaussian distribution for the error terms is investigated, for example, in \citet{kmenta1968,oberhofer1974,magnus1978,park1993}. Further developments have been obtained by using bootstrap methods \citep[see, e.g.,][]{rocke1989,rilstone1996} and a likelihood distributional analysis \citep{fraser2005}. Many studies have been performed also in a Bayesian framework \citep[see, e.g.,][]{zellner1971,percy1992,ando2010,zellner2010a}. Most of these methods have been developed under the assumption that the distribution of the error terms is Gaussian. Properties of the feasible generalized least squares estimators under non-Gaussian errors or solutions obtained using other distributions are described, for example, in \citet{srivastava1995,kurata1999,kowalski1999,ng2002,zellner2010b}.

The aim of this paper is to propose the use of finite mixtures for modelling the error term distribution in a seemingly unrelated linear regression model. Finite mixture models are widely employed in many areas of multivariate analysis, especially for model-based cluster analysis, discriminant analysis and multivariate density estimation \citep[see, e.g.,][]{mclachlan2000}. Recently, finite mixtures of Gaussian and Student-$t$ distributions have been employed also in multiple and multivariate linear regression analysis \citep[see, e.g.,][]{bartolucci2005,soffritti2011,galimberti2013} to handle non-normal error terms. This approach has the advantage of capturing the effect of omitted nominal regressors from the model and obtaining robust estimates of the regression coefficients when the distribution of the error terms is non-normal. In this paper the same approach is applied to the seemingly unrelated regression model. In particular, the focus is on seemingly unrelated linear regression models in which the error terms are assumed to follow a finite mixture of multivariate Gaussian distributions.

The paper is organized as follows. Section~\ref{sec:methods} illustrates the theory behind the new methodology. Namely, the novel models are presented in Section~\ref{sec:model}. Theorem~\ref{th:identificabilita} provides conditions for the model identifiability (Section~\ref{sec:identificabilita}). The score vector and the Hessian matrix for the model parameter are reported in Section~\ref{sec:score} (Theorems~\ref{theo:score} and \ref{theo:hessian}). Details about the maximum likelihood (ML) estimation through an Expectation-Maximisation (EM) algorithm are given in Section~\ref{sec:em}.
Results obtained from the analysis of a real dataset using the proposed approach and other
methods are presented in Section~\ref{sec:experimental}. In Section~\ref{sec:conclusion} some concluding remarks are provided. Proofs of Theorems~\ref{theo:score} and \ref{theo:hessian} and other technical results are in Appendix.

\section{Seemingly unrelated regression models with a mixture of Gaussian components for the error terms}\label{sec:methods}
\subsection{The general model}\label{sec:model}
The novel model can be introduced as follows. Let $\mathbf{Y}_i=(Y_{i1}, \ldots, Y_{id}, \ldots, Y_{iD})'$ be the vector of the $D$ dependent variables for the $i$th observation, $i=1, \ldots, I$. Furthermore, let $\mathbf{x}_{id}$ be the vector composed of the fixed values of the $P_d$ regressors for the $i$th observation in the equation for the $d$th dependent variable, $d=1, \ldots, D$. A seemingly unrelated regression model can be defined through the following system of equations:
\begin{equation}\label{eq:modelloyivecchiosistema}
\left\{
\begin{array}{l}
  Y_{i1}  = \beta_{01}  + \mathbf{x}_{i1}'\boldsymbol{\beta}_1  + \epsilon_{i1} \\
  \vdots   \\
  Y_{id}  = \beta_{0d}  + \mathbf{x}_{id}'\boldsymbol{\beta}_d  + \epsilon_{id} \\
  \vdots \\
  Y_{iD}  = \beta_{0D}  + \mathbf{x}_{iD}'\boldsymbol{\beta}_D + \epsilon_{iD}
\end{array}
\right. \ \ i = 1, \ldots, I,
\end{equation}
where $\beta_{0d}$, $\boldsymbol{\beta}_d$, and $\epsilon_{id}$ are the intercept, the regression coefficient vector and the error term for the $i$th observation in the equation for the $d$th dependent variable, respectively.
Equation (\ref{eq:modelloyivecchiosistema}) can be written in compact form using the following matrix notation:
\begin{equation}\label{eq:modelloyivecchio}
\mathbf{Y}_i= \boldsymbol{\beta}_0 + \mathbf{X}_{i}'\boldsymbol{\beta}+\boldsymbol{\epsilon}_{i}, \ \ i = 1, \ldots, I,
\end{equation}
where $\boldsymbol{\beta}_0=\left(\beta_{01}, \ldots, \beta_{0d}, \ldots, \beta_{0D}\right)'$,
 $\boldsymbol{\beta}=\left(\boldsymbol{\beta}'_1, \ldots, \boldsymbol{\beta}'_d, \ldots, \boldsymbol{\beta}'_D\right)'$, $\boldsymbol{\epsilon}_{i}=(\epsilon_{i1}, \ldots, \epsilon_{id}, \ldots, \epsilon_{iD})'$, and $\mathbf{X}_{i}$ is the following $P \times D$ partitioned matrix:\\
\begin{equation}\label{eq:matriceXi}
\left[
  \begin{array}{cccc}
    \mathbf{x}_{i1} & \mathbf{0}_{P_1} & \cdots & \mathbf{0}_{P_1} \\
    \mathbf{0}_{P_2} & \mathbf{x}_{i2} & \cdots & \mathbf{0}_{P_2} \\
    \vdots & \vdots &  & \vdots \\
    \mathbf{0}_{P_D} & \mathbf{0}_{P_D} & \cdots & \mathbf{x}_{iD} \\
  \end{array}
\right],
\end{equation}
with $\mathbf{0}_{P_d}$ denoting the $P_d$-dimensional null vector and $P=\sum_{d=1}^D P_d$.
\begin{rem}
Note that this definition of seemingly unrelated regression model differs from the one originally introduced by \citet{zellner1962}; however, these two definitions are equivalent (see, for example, \citet{park1993}). The choice of the model definition given in equation (\ref{eq:modelloyivecchio}) is motivated by its analytical convenience in deriving some technical results described in this paper.
\end{rem}

The proposed model is based on the assumption that the $I$ error terms are independent and identically distributed, and that
\begin{equation} \label{eq:mistura}
\boldsymbol{\epsilon}_i \sim \sum_{k=1}^K \pi_k N_D(\boldsymbol{\nu}_{k},\boldsymbol{\Sigma}_{k}), \ \ i = 1, \ldots, I,
\end{equation}
where $\pi_k$'s are positive weights that sum to 1, the $\boldsymbol{\nu}_{k}$'s are $D$-dimensional mean vectors that satisfy the constraint $\sum_{k=1}^K \pi_k \boldsymbol{\nu}_{k}= \boldsymbol{0}_D$, the $\boldsymbol{\Sigma}_{k}$'s are $D \times D$ positive definite symmetric matrices and $N_D(\boldsymbol{\nu}_{k},\boldsymbol{\Sigma}_{k})$ denotes the $D$-dimensional Gaussian distribution with parameters $\boldsymbol{\nu}_{k}$ and $\boldsymbol{\Sigma}_{k}$.

Given equations (\ref{eq:modelloyivecchio}) and ($\ref{eq:mistura}$), the probability density function (p.d.f.) of the $D$-dimensional random vector $\mathbf{Y}_i$ is
\begin{equation} \label{eq:mixreg}
\sum_{k=1}^K \pi_k \phi_D(\mathbf{y}_i; \boldsymbol{\lambda}_{k}+ \mathbf{X}_{i}'\boldsymbol{\beta}, \boldsymbol{\Sigma}_{k}), \ \ \mathbf{y}_i \in \mathbb{R}^D, \ \ i=1, \ldots, I,
\end{equation}
where $\phi_D(\mathbf{y}_i;\boldsymbol{\mu},\boldsymbol{\Sigma})$ is the p.d.f. of the $D$-dimensional Gaussian distribution $N_D(\boldsymbol{\mu},\boldsymbol{\Sigma})$ evaluated at $\mathbf{y}_i$, and $\boldsymbol{\lambda}_{k}=\boldsymbol{\beta}_0+\boldsymbol{\nu}_{k}$. Differently from the $\boldsymbol{\nu}_{k}$'s, the $\boldsymbol{\lambda}_{k}$'s are not subject to any constraint. For this reason, in this paper the attention is focused on the vector of the model parameters given by
$\boldsymbol{\theta}=(\boldsymbol{\pi}', \boldsymbol{\beta}',
\boldsymbol{\theta}'_1, \ldots, \boldsymbol{\theta}'_K)'$, where $\boldsymbol{\pi}=(\pi_1, \ldots, \pi_{K-1})'$,
$\boldsymbol{\theta}_k =\left(\boldsymbol{\lambda}'_k, \mathrm{v}\left(\boldsymbol{\Sigma}_k\right)'\right)'$ for $k=1, \ldots, K$, with $\mathrm{v}(\boldsymbol{\Sigma}_{k})$ denoting the $\frac{1}{2}D(D+1)$-dimensional vector formed by stacking the columns of the lower triangular portion of $\boldsymbol{\Sigma}_{k}$ \citep[see, e.g.,][]{schott2005}.

Suppose that the $i$th observation was drawn from the $k$th component of the mixture. Then, the equation for such an observation would be
\begin{equation}\label{eq:modelloyinuovo}
\mathbf{Y}_i=\boldsymbol{\lambda}_k+\mathbf{X}_{i}'\boldsymbol{\beta}+\tilde{\boldsymbol{\epsilon}}_{ik},
\end{equation}
where $\tilde{\boldsymbol{\epsilon}}_{ik} \sim N_D(\boldsymbol{0}_D,\boldsymbol{\Sigma}_k)$.
The model defined by equation (\ref{eq:mixreg}) can be seen as a mixture of $K$ seemingly unrelated linear regression models with Gaussian error terms. In this model, observations drawn from different components have different intercepts for the $D$ dependent variables and different covariance matrices for the error terms, while the regression coefficients are equal across components. In the special case where $K=1$, this model results in the classical seemingly unrelated regression model with Gaussian errors. If $\mathbf{x}_{id}=\mathbf{x}_{i}$ $\forall d$ (the vectors of the regressors for the $D$ equations coincide), the model proposed by \citet{soffritti2011} is obtained. Furthermore, the model proposed by \citet{bartolucci2005} can be obtained when $D=1$. Finally, if $P_d=0$ $\forall d$, model (\ref{eq:mixreg}) results in the mixture model with $K$ Gaussian components \citep[see, e.g.,][]{mclachlan2000}.

\subsection{Model identifiability}\label{sec:identificabilita}

As any finite mixture model, also model (\ref{eq:mixreg}) is invariant under permutations of the labels of the $K$ components \citep[see, e.g.,][]{mclachlan2000}. For the proposed model, whose parameter is
$\boldsymbol{\theta}=(\boldsymbol{\pi}', \boldsymbol{\beta}',
\boldsymbol{\theta}'_1, \ldots, \boldsymbol{\theta}'_K)'$, the following theorem holds:

\begin{thm}\label{th:identificabilita}
The linear regression model (\ref{eq:mixreg}) is identifiable, provided that, for $d=1, \ldots, D$, vectors
$\{\mathbf{x}_{id}, i=1, \ldots, I\}$ do not lie on a common $(P_d-1)$-dimensional hyperplane.
\end{thm}

\begin{proof}
The identifiability condition described in Theorem~\ref{th:identificabilita} is a generalization of the usual condition for the identifiability of a multiple linear regression model. It is required in order to guarantee identifiability of the parameters $\boldsymbol{\beta}$ and $\boldsymbol{\lambda}_1, \ldots, \boldsymbol{\lambda}_K$ that characterize the conditional expectations for the $D$ dependent variables.

Furthermore, consider the joint conditional p.d.f. of a random sample $\mathbf{y}_1, \ldots, \mathbf{y}_i, \ldots, \mathbf{y}_I$ from the model (\ref{eq:mixreg}), given the fixed values of the regressors contained in $\mathbf{X}_{1}, \ldots, \mathbf{X}_{I}$:

\begin{equation} \label{eq:jointpdf}
f(\mathbf{y}_1, \ldots, \mathbf{y}_I; \mathbf{X}_{1}, \ldots, \mathbf{X}_{I}, \boldsymbol{\theta})=\prod_{i=1}^I \left[\sum_{k=1}^K \pi_k \phi_D(\mathbf{y}_i; \boldsymbol{\lambda}_{k}+ \mathbf{X}_{i}'\boldsymbol{\beta}, \boldsymbol{\Sigma}_{k})\right].
\end{equation}

It is possible to show that (\ref{eq:jointpdf}) can be written as the following mixture of $J$ Gaussian components:
\begin{equation} \label{eq:jointpdf2}
f(\mathbf{y}_1, \ldots, \mathbf{y}_I; \mathbf{X}_{1}, \ldots, \mathbf{X}_{I}, \boldsymbol{\theta})= \sum_{j=1}^J \pi_j \phi_{D \cdot I}(\mathbf{y}; \boldsymbol{\lambda}_{j}+ \mathbf{X}\boldsymbol{\beta}, \boldsymbol{\Sigma}_{j}),
\end{equation}
where $J=K^I$, $\mathbf{y}=\left(\mathbf{y}'_1, \ldots, \mathbf{y}'_i, \ldots, \mathbf{y}'_I\right)'$,
$\mathbf{X}=\left[\mathbf{X}_{1} \ldots \mathbf{X}_{i} \ldots \mathbf{X}_{I}\right]'$, $\pi_j=\prod_{i=1}^I  \pi_{k_i^{(j)}}$, $\boldsymbol{\lambda}_j = (\boldsymbol{\lambda}_{k_1^{(j)}}', \ldots, \boldsymbol{\lambda}_{k_i^{(j)}}', \ldots, \boldsymbol{\lambda}_{k_I^{(j)}}')'$, $\boldsymbol{\Sigma}_j=$ $diag(\boldsymbol{\Sigma}_{k_1^{(j)}},$ $\ldots, \boldsymbol{\Sigma}_{k_i^{(j)}}, \ldots, \boldsymbol{\Sigma}_{k_I^{(j)}})$ is a block diagonal matrix, and $\textbf{k}^{(j)}=\left(k_1^{(j)}, \ldots, k_I^{(j)}\right)'$ is the $j$th element of the set $A_{K,I}=\{(k_1, \ldots, k_I)': k_i \in \{1, \ldots, K\}, \ i=1, \ldots, I\}$ containing the $J$ arrangements of the first $K$ positive integers amongst $I$ with repetitions.
The proof can be completed by showing that mixtures (\ref{eq:jointpdf2}) are identifiable. The proof of this latter result can be found in \citet{soffritti2011}.
\end{proof}

\subsection{Score vector and Hessian matrix}\label{sec:score}

Given a random sample $\mathbf{y}_1, \ldots, \mathbf{y}_i, \ldots, \mathbf{y}_I$ from the model (\ref{eq:mixreg}), the
log-likelihood is
\begin{equation}\label{eq:loglik}
l(\boldsymbol{\theta})=\sum_{i=1}^I \ln \left(\sum_{k=1}^K \pi_k \phi_D(\mathbf{y}_i; \boldsymbol{\lambda}_{k}+ \mathbf{X}_{i}'\boldsymbol{\beta}, \boldsymbol{\Sigma}_{k})\right).
\end{equation}
The log-likelihood (\ref{eq:loglik}) can be used to derive the ML estimator of $\boldsymbol{\theta}$. Furthermore, \citet{redner1984} showed that, under suitable conditions, an estimate of the asymptotic variance of the ML estimator of the parameters in a finite mixture model can be obtained using the Hessian matrix. In order to obtain the score vector and the Hessian matrix the following notation is introduced.
Let
\begin{equation}\nonumber
f_{ki}  = \frac{\pi_k}{(2\pi)^{D/2}\det\left( \boldsymbol{\Sigma}_k\right)^{1/2}}\exp\left[ -\frac{1}{2}\left(\mathbf{y}_i-\boldsymbol{\lambda}_k-\mathbf{X}'_{i}\boldsymbol{\beta} \right)'\boldsymbol{\Sigma}_k^{-1}\left(\mathbf{y}_i-\boldsymbol{\lambda}_k-\mathbf{X}'_{i}\boldsymbol{\beta} \right) \right];
\end{equation}
$\alpha_{ki}=\frac{f_{ki}}{\left(\sum_{l=1}^K f_{li}\right)}$; $\mathbf{a}_k=\frac{1}{\pi_k}\mathbf{e}_k$ for $k=1,\ldots,K-1$ and $\mathbf{a}_K=-\frac{1}{\pi_K}\boldsymbol{1}_{(K-1)}$, where $\mathbf{e}_k$ is the $k$th column of $\mathbf{I}_{(K-1)}$ (the identity matrix of order $K-1$) and $\boldsymbol{1}_{(K-1)}$ denotes the $(K-1)$-dimensional vector having each component equal to 1;
$\mathbf{b}_{ki}=\boldsymbol{\Sigma}_k^{-1}\left(\mathbf{y}_i-\boldsymbol{\lambda}_k-\mathbf{X}_{i}'\boldsymbol{\beta} \right)$; $\mathbf{B}_{ki}=\boldsymbol{\Sigma}_k^{-1}-\mathbf{b}_{ki}\mathbf{b}'_{ki}$;
\begin{equation}\nonumber
\mathbf{c}_{ki}=\left[\begin{array}{c}
\mathbf{b}_{ki} \\
-\frac{1}{2} \mathbf{G}' \mathrm{vec}\left(\mathbf{B}_{ki}\right)
\end{array}\right],
\end{equation}
where $\mathbf{G}$ denotes the duplication matrix and $\mathrm{vec}(\mathbf{B}_{ki})$ denotes the vector formed by stacking the columns of the matrix $\mathbf{B}_{ki}$ one underneath the other \citep[see, e.g.,][]{schott2005}.

\begin{thm}\label{theo:score}
The score vector for the parameters of model (\ref{eq:mixreg}) is composed of the sub-vectors
$\frac{\partial}{\partial\boldsymbol{\pi}'}l\left(\boldsymbol{\theta}\right)$, $\frac{\partial}{\partial\boldsymbol{\beta}'}l\left(\boldsymbol{\theta}\right)$, $\frac{\partial}{\partial\boldsymbol{\theta}'_1}l\left(\boldsymbol{\theta}\right), \ldots, \frac{\partial}{\partial\boldsymbol{\theta}'_K}l\left(\boldsymbol{\theta}\right)$,
where
\begin{eqnarray} \nonumber
\frac{\partial}{\partial\boldsymbol{\pi}}l\left(\boldsymbol{\theta}\right) & = & \sum_{i=1}^I\bar{\mathbf{a}}_i,\\ \nonumber
\frac{\partial}{\partial\boldsymbol{\beta}}l\left(\boldsymbol{\theta}\right) & = & \sum_{i=1}^I\mathbf{X}_{i}\bar{\mathbf{b}}_i, \\ \nonumber
\frac{\partial}{\partial\boldsymbol{\theta}_k}l\left(\boldsymbol{\theta}\right) & = & \sum_{i=1}^I \alpha_{ki}\mathbf{c}_{ki}, \ \ k=1, \ldots, K,
\end{eqnarray}
with $\bar{\mathbf{a}}_i=\sum_{k=1}^K \alpha_{ki}\mathbf{a}_k$ and $\bar{\mathbf{b}}_i=\sum_{k=1}^K \alpha_{ki}\mathbf{b}_{ki}$.
\end{thm}

\begin{thm}\label{theo:hessian}
The Hessian matrix $H(\boldsymbol{\theta})$ for the parameters of model (\ref{eq:mixreg}) is equal to
\begin{equation}\label{eq:hessian}
\left[
  \begin{array}{ccccc}
     \frac{\partial^2}{\partial\boldsymbol{\pi}\partial\boldsymbol{\pi}'}l\left(\boldsymbol{\theta}\right) &
     \frac{\partial^2}{\partial\boldsymbol{\pi}\partial\boldsymbol{\beta}'}l\left(\boldsymbol{\theta}\right)&
     \frac{\partial^2}{\partial\boldsymbol{\pi}\partial\boldsymbol{\theta}_1'}l\left(\boldsymbol{\theta}\right)&
     \cdots & \frac{\partial^2}{\partial\boldsymbol{\pi}\partial\boldsymbol{\theta}_K'}l\left(\boldsymbol{\theta}\right) \\
   \frac{\partial^2}{\partial\boldsymbol{\beta}\partial\boldsymbol{\pi}'}l\left(\boldsymbol{\theta}\right) &
   \frac{\partial^2}{\partial\boldsymbol{\beta}\partial\boldsymbol{\beta}'}l\left(\boldsymbol{\theta}\right)&
   \frac{\partial^2}{\partial\boldsymbol{\beta}\partial\boldsymbol{\theta}_1'}l\left(\boldsymbol{\theta}\right)&
   \cdots & \frac{\partial^2}{\partial\boldsymbol{\beta}\partial\boldsymbol{\theta}_K'}l\left(\boldsymbol{\theta}\right) \\
    \frac{\partial^2}{\partial\boldsymbol{\theta}_1\partial\boldsymbol{\pi}'}l\left(\boldsymbol{\theta}\right) &
    \frac{\partial^2}{\partial\boldsymbol{\theta}_1\partial\boldsymbol{\beta}'}l\left(\boldsymbol{\theta}\right)&
    \frac{\partial^2}{\partial\boldsymbol{\theta}_1\partial\boldsymbol{\theta}_1'}l\left(\boldsymbol{\theta}\right)&
    \cdots & \frac{\partial^2}{\partial\boldsymbol{\theta}_1\partial\boldsymbol{\theta}_K'}l\left(\boldsymbol{\theta}\right) \\      \cdots & \cdots & \cdots & \cdots & \cdots \\
   \frac{\partial^2}{\partial\boldsymbol{\theta}_K\partial\boldsymbol{\pi}'}l\left(\boldsymbol{\theta}\right) &
   \frac{\partial^2}{\partial\boldsymbol{\theta}_K\partial\boldsymbol{\beta}'}l\left(\boldsymbol{\theta}\right)&
   \frac{\partial^2}{\partial\boldsymbol{\theta}_K\partial\boldsymbol{\theta}_1'}l\left(\boldsymbol{\theta}\right)&
   \cdots & \frac{\partial^2}{\partial\boldsymbol{\theta}_K\partial\boldsymbol{\theta}_K'}l\left(\boldsymbol{\theta}\right)      \end{array}
\right],
\end{equation}
where
\begin{eqnarray}\nonumber
\frac{\partial^2}{\partial\boldsymbol{\pi}\partial\boldsymbol{\pi}'}l\left(\boldsymbol{\theta}\right) & = &
-\sum_{i=1}^I \bar{\mathbf{a}}_i\bar{\mathbf{a}}'_i,\\ \nonumber
\frac{\partial^2}{\partial\boldsymbol{\pi}\partial\boldsymbol{\beta}'}l\left(\boldsymbol{\theta}\right) & = &
\sum_{i=1}^I \left[\left(\sum_{k=1}^K\alpha_{ki}\mathbf{a}_k \mathbf{b}'_{ki}\right)-\bar{\mathbf{a}}_i \bar{\mathbf{b}}'_{i}\right]\mathbf{X}'_{i},\\ \nonumber
\frac{\partial^2}{\partial\boldsymbol{\pi}\partial\boldsymbol{\theta}_k'}l\left(\boldsymbol{\theta}\right) & = &
\sum_{i=1}^I \alpha_{ki}\left(\mathbf{a}_{k}-\bar{\mathbf{a}}_i\right)\mathbf{c}'_{ki},  \ \ k=1, \ldots, K,\\ \nonumber
\frac{\partial^2}{\partial\boldsymbol{\beta}\partial\boldsymbol{\beta}'}l\left(\boldsymbol{\theta}\right) & = &
- \sum_{i=1}^I \mathbf{X}_{i}\left[\bar{\mathbf{B}}_{i}+\bar{\mathbf{b}}_{i}\bar{\mathbf{b}}'_{i}\right]\mathbf{X}'_{i},\\ \nonumber
 \frac{\partial^2}{\partial\boldsymbol{\beta}\partial\boldsymbol{\theta}_k'}l\left(\boldsymbol{\theta}\right) & = &
- \sum_{i=1}^I    \alpha_{ki}\mathbf{X}_{i}\left[\mathbf{F}_{ki}-\left(\mathbf{b}_{ki}-\bar{\mathbf{b}}_{i}\right)\mathbf{c}'_{ki}\right], \ \ k=1, \ldots, K,\\ \nonumber \frac{\partial^2}{\partial\boldsymbol{\theta}_k\partial\boldsymbol{\theta}_k'}l\left(\boldsymbol{\theta}\right) & = &
- \sum_{i=1}^I
\alpha_{ki}\left[\mathbf{C}_{ki}-\left(1-\alpha_{ki}\right)\mathbf{c}_{ki}\mathbf{c}'_{ki}\right], \ \ k=1, \ldots, K,\\ \nonumber
\frac{\partial^2}{\partial\boldsymbol{\theta}_k\partial\boldsymbol{\theta}'_h}l\left(\boldsymbol{\theta}\right) & = &
- \sum_{i=1}^I
\alpha_{ki}\alpha_{hi}\mathbf{c}_{ki}\mathbf{c}'_{hi}, \ \ \forall k \neq h,
\end{eqnarray}
with $\bar{\mathbf{B}}_{i}=\sum_{k=1}^K\alpha_{ki}\left(\boldsymbol{\Sigma}_k^{-1}-\mathbf{b}_{ki}\mathbf{b}'_{ki}\right)$,
$\mathbf{F}_{ki}=\left[\begin{array}{cc}
\boldsymbol{\Sigma}_k^{-1} & \left(\mathbf{b}'_{ki}\otimes \boldsymbol{\Sigma}_k^{-1}\right)\mathbf{G}
\end{array}\right]$ and
\begin{equation}\nonumber
\mathbf{C}_{ki}=\left[\begin{array}{cc}
\boldsymbol{\Sigma}_k^{-1} & \left(\mathbf{b}'_{ki}\otimes \boldsymbol{\Sigma}_k^{-1}\right)\mathbf{G}\\
 \mathbf{G}'\left(\mathbf{b}_{ki}\otimes \boldsymbol{\Sigma}_k^{-1}\right) & \frac{1}{2}\mathbf{G}'\left[\left(\boldsymbol{\Sigma}_k^{-1}-2\mathbf{B}_{ki}\right)\otimes \boldsymbol{\Sigma}_k^{-1}\right]\mathbf{G}\end{array}\right].
\end{equation}
\end{thm}
Proofs of Theorems~\ref{theo:score} and \ref{theo:hessian} are provided in \ref{proof:score} and \ref{proof:hessian}, respectively.

\begin{rem}
Theorems~\ref{theo:score} and \ref{theo:hessian} provide the score vector and the Hessian matrix not only for the model proposed in this paper, but also for the models introduced in \citet{bartolucci2005} and \citet{soffritti2011}, after some suitable simplifications. Furthermore, they represent a generalization of Theorem 1 in \citet{boldea2009}.
\end{rem}

\subsection{An EM algorithm for maximum likelihood estimation}\label{sec:em}

The score vector and the Hessian matrix described in Section~\ref{sec:score} can be used to compute the ML estimates of the model parameter $\boldsymbol{\theta}$ through a Newton-Raphson algorithm for the maximisation of $l(\boldsymbol{\theta})$ in equation (\ref{eq:loglik}). However, the evaluation of the Hessian matrix at each iteration can be computationally expensive, especially with large samples. In order to avoid this problem, in this Section an EM algorithm is developed, using the approach for incomplete-data problems \citep{dempster1977,mclachlan2008}. This approach is widely employed in finite mixture models, where the source of unobservable information is the specific component of the mixture model that generates each sample observation. Specifically, this unobservable information for the $i$th observation can be described by the $K$-dimensional vector $\mathbf{z}'_i=\left(z_{i1}, \ldots, z_{iK}\right)$, where
$z_{ik}=1$ when $\mathbf{y}_i$ is generated from the $k$th component, and $z_{ik}=0$ otherwise, for $k=1, \ldots, K$.
Thus, $\sum_{k=1}^K z_{ik}=1$, $i=1, \ldots, I$.

Consider the following hierarchical representation for $\mathbf{y}_i|\mathbf{X}_i$:
\begin{equation}\nonumber
  \mathbf{z}_i  \sim mult(1, \pi_1, \ldots, \pi_K),
\end{equation}
\begin{equation}\nonumber
  \mathbf{y}_i|(\mathbf{X}_i, z_{ik}=1) \sim N_D\left(\boldsymbol{\lambda}_k+ \mathbf{X}'_i \boldsymbol{\beta}, \boldsymbol{\Sigma}_k\right),
\end{equation}
where $mult(1, \pi_1, \ldots, \pi_K)$ denotes the $K$-dimensional multinomial distribution with parameters $\pi_1, \ldots, \pi_K$, and assume that this representation independently holds for $i=1, \ldots, I$. Then, the complete-data log-likelihood $l_c(\boldsymbol{\theta})$ of model (\ref{eq:mixreg}) can be expressed as
\begin{equation}\label{eq:lc}
l_c\left(\boldsymbol{\theta}\right)=\sum_{i=1}^I \sum_{k=1}^K z_{ik}\ln f_{ki}.
\end{equation}

The first order differential of $l_c(\boldsymbol{\theta})$ is
\begin{eqnarray}\nonumber
\mathrm{d}l_c\left(\boldsymbol{\theta}\right) & = & \sum_{i=1}^I \sum_{k=1}^K z_{ik} \mathrm{d}\ln f_{ki}\\ \nonumber
& = & \left(\mathrm{d}\boldsymbol{\pi}\right)'\sum_{k=1}^K z_{\cdot k}\mathbf{a}_k + \left(\mathrm{d}\boldsymbol{\beta}\right)'\sum_{i=1}^I \sum_{k=1}^K z_{ik}\mathbf{X}_{i}\mathbf{b}_{ki}
+\sum_{k=1}^K \left(\mathrm{d}\boldsymbol{\theta}_k\right)'\sum_{i=1}^I z_{ik}\mathbf{c}_{ki}\\ \nonumber
& =& \left(\mathrm{d}\boldsymbol{\pi}\right)'\sum_{k=1}^K z_{\cdot k}\mathbf{a}_k \\ \label{eq:dlc2}
& & + \sum_{i=1}^I \sum_{k=1}^K z_{ik}\left[\left(\mathrm{d}\boldsymbol{\lambda}_k\right)'+ \left(\mathrm{d}\boldsymbol{\beta}\right)'\mathbf{X}_{i}\right]\mathbf{b}_{ki}\\ \label{eq:dlc3}
& & -\frac{1}{2}\sum_{k=1}^K \mathrm{d}\left(\mathrm{v}\boldsymbol{\Sigma}_k\right)'\mathbf{G}'\mathrm{vec}\left(\sum_{i=1}^I z_{ik}\mathbf{B}_{ki}\right)
\end{eqnarray}
where the second and third equalities are obtained using equation (\ref{eq:dlnfki2}) in \ref{proof:score}, and $z_{\cdot k}=\sum_{i=1}^I z_{ik}$.

To determine the solution of each M step of the EM algorithm, it is convenient to introduce the following notation.
Let $\mathrm{d}l_{c2}$ and $\mathrm{d}l_{c3}$ denote the expressions in equations (\ref{eq:dlc2}) and (\ref{eq:dlc3}), respectively. Let $\boldsymbol{\gamma}=\left(\boldsymbol{\lambda}'_{1}, \ldots, \boldsymbol{\lambda}'_{K}, \boldsymbol{\beta}'\right)'$ be the $(D \cdot K +P)$-dimensional vector comprising the intercepts of all components and regression coefficients for all dependent variables. $\mathbf{O}_k$ is a matrix of dimension $(D\cdot K)\times D$ obtained extracting the columns of the matrix $\mathbf{I}_{(D\cdot K)}$ from the $(1+(k-1)\cdot D)$th to the $(D+(k-1)\cdot D)$th, for $k=1,\ldots, K$. Furthermore, let
$\mathbf{X}_{ki}=\left[ \begin{array}{c}
\mathbf{O}_k\\
\mathbf{X}_i
\end{array}\right]$; this is a matrix of dimension $\left(D\cdot K+P\right)\times D$ such that $\mathbf{X}'_{ki}\boldsymbol{\gamma}=\boldsymbol{\lambda}_k+\mathbf{X}'_{i}\boldsymbol{\beta}$ and
$\mathbf{X}'_{ki}\mathrm{d}\boldsymbol{\gamma}=\mathrm{d}\boldsymbol{\lambda}_k+\mathbf{X}'_{i}\mathrm{d}\boldsymbol{\beta}$.
Using this latter notation, the expressions of $\mathrm{d}l_{c2}$ and $\mathrm{d}l_{c3}$ in equations (\ref{eq:dlc2}) and (\ref{eq:dlc3}) turn into
\begin{eqnarray}\nonumber
\mathrm{d}l_{c2} & = & \sum_{i=1}^I \sum_{k=1}^K z_{ik}\left(\mathrm{d}\boldsymbol{\gamma}\right)'\mathbf{X}_{ki}\mathbf{b}_{ki}\\ \nonumber
 & = & \left(\mathrm{d}\boldsymbol{\gamma}\right)'\sum_{i=1}^I \sum_{k=1}^K  z_{ik}\mathbf{X}_{ki}\boldsymbol{\Sigma}_{k}^{-1}\left(\mathbf{y}_i-\mathbf{X}'_{ki}\boldsymbol{\gamma}\right)\\ \nonumber
 & = & \left(\mathrm{d}\boldsymbol{\gamma}\right)'\sum_{i=1}^I \sum_{k=1}^K  z_{ik}\mathbf{X}_{ki}\boldsymbol{\Sigma}_{k}^{-1}\mathbf{y}_i\\  \label{eq:dlc2b}
& & -\left(\mathrm{d}\boldsymbol{\gamma}\right)'\left(\sum_{i=1}^I \sum_{k=1}^K z_{ik}\mathbf{X}_{ki}\boldsymbol{\Sigma}_{k}^{-1}\mathbf{X}'_{ki}\right)\boldsymbol{\gamma},
\end{eqnarray}
\begin{eqnarray}\nonumber
\mathrm{d}l_{c3} & = & -\frac{1}{2}\sum_{k=1}^K \mathrm{d}\left(\mathrm{v}\boldsymbol{\Sigma}_k\right)'\mathbf{G}^{\top} \mathrm{vec}\left(\sum_{i=1}^I z_{ik}\boldsymbol{\Sigma}_k^{-1}-\sum_{i=1}^I z_{ik}\mathbf{b}_{ki}\mathbf{b}'_{ki}\right) \\ \label{eq:dlc3b}
& = & -\frac{1}{2}\sum_{k=1}^K \mathrm{d}\left(\mathrm{v}\boldsymbol{\Sigma}_k\right)'\mathbf{G}' \mathrm{vec}\left(z_{\cdot k}\boldsymbol{\Sigma}_k^{-1}-\boldsymbol{\Sigma}_k^{-1}\mathbf{S}_k\boldsymbol{\Sigma}_k^{-1}\right),
\end{eqnarray}
where $\mathbf{S}_k=\sum_{i=1}^I z_{ik}\left(\mathbf{y}_i-\mathbf{X}'_{ki}\boldsymbol{\gamma}\right)\left(\mathbf{y}_i-\mathbf{X}'_{ki}\boldsymbol{\gamma}\right)'$.
Using equation (\ref{eq:dlc3b}) and some properties of the vec operator \citep[see, in particular,][Theorem 8.11]{schott2005}, it is also possible to write
\begin{eqnarray}\nonumber
\mathrm{d}l_{c3} & = & \frac{1}{2}\sum_{k=1}^K  \mathrm{d}\left(\mathrm{v}\boldsymbol{\Sigma}_k\right)'\mathbf{G}' \mathrm{vec}\left[\boldsymbol{\Sigma}_k^{-1}\left(\mathbf{S}_k-z_{\cdot k}\boldsymbol{\Sigma}_k\right)\boldsymbol{\Sigma}_k^{-1}\right] \\ \label{eq:dlc3c}
& = & \frac{1}{2}\sum_{k=1}^K \mathrm{d}\left(\mathrm{v}\boldsymbol{\Sigma}_k\right)'\mathbf{G}' \left(\boldsymbol{\Sigma}_k^{-1}\otimes \boldsymbol{\Sigma}_k^{-1}\right) \mathbf{G}\mathrm{v}\left(\mathbf{S}_k-z_{\cdot k}\boldsymbol{\Sigma}_k\right).
\end{eqnarray}
Thus, the following alternative expression for $\mathrm{d}l_c\left(\boldsymbol{\theta}\right)$ holds:
\begin{eqnarray}\nonumber
\mathrm{d}l_c\left(\boldsymbol{\theta}\right) & = & \left(\mathrm{d}\boldsymbol{\pi}\right)'\sum_{k=1}^K z_{\cdot k}\mathbf{a}_k+ \left(\mathrm{d}\boldsymbol{\gamma}\right)'\sum_{i=1}^I \sum_{k=1}^K  z_{ik}\mathbf{X}_{ki}\boldsymbol{\Sigma}_{k}^{-1}\mathbf{y}_i\\ \nonumber
& & -\left(\mathrm{d}\boldsymbol{\gamma}\right)'\left(\sum_{i=1}^I \sum_{k=1}^K z_{ik}\mathbf{X}_{ki}\boldsymbol{\Sigma}_{k}^{-1}\mathbf{X}'_{ki}\right)\boldsymbol{\gamma}\\ \label{eq:dlc}
 & &+ \frac{1}{2}\sum_{k=1}^K \mathrm{d}\left(\mathrm{v}\boldsymbol{\Sigma}_k\right)'\mathbf{G}' \left(\boldsymbol{\Sigma}_k^{-1}\otimes \boldsymbol{\Sigma}_k^{-1}\right) \mathbf{G}\mathrm{v}\left(\mathbf{S}_k-z_{\cdot k}\boldsymbol{\Sigma}_k\right).
\end{eqnarray}

The first derivatives of $l_c\left(\boldsymbol{\theta}\right)$ with respect to the parameters $\boldsymbol{\pi}$,  $\boldsymbol{\gamma}$ and $\mathrm{v}\boldsymbol{\Sigma}_k$ ($k=1, \ldots, K)$ are:
\begin{eqnarray}\nonumber
\frac{\partial}{\partial\boldsymbol{\pi}}l_c\left(\boldsymbol{\theta}\right) & = & \sum_{k=1}^K z_{\cdot k}\mathbf{a}_k,\\ \nonumber
\frac{\partial}{\partial\boldsymbol{\gamma}}l_c\left(\boldsymbol{\theta}\right) & = &
\sum_{i=1}^I \sum_{k=1}^K  z_{ik}\mathbf{X}_{ki}\boldsymbol{\Sigma}_{k}^{-1}\mathbf{y}_i-
\left(\sum_{i=1}^I \sum_{k=1}^K z_{ik}\mathbf{X}_{ki}\boldsymbol{\Sigma}_{k}^{-1}\mathbf{X}'_{ki}\right)\boldsymbol{\gamma},\\ \nonumber
\frac{\partial}{\partial\left(\mathrm{v}\boldsymbol{\Sigma}_k\right)}l_c\left(\boldsymbol{\theta}\right) & = &
\frac{1}{2}\mathbf{G}' \left(\boldsymbol{\Sigma}_k^{-1}\otimes \boldsymbol{\Sigma}_k^{-1}\right) \mathbf{G}\mathrm{v}\left(\mathbf{S}_k-z_{\cdot k}\boldsymbol{\Sigma}_k\right), \ \ k=1, \ldots, K.
\end{eqnarray}

In order to maximise $l_c\left(\boldsymbol{\theta}\right)$ these derivatives are set equal to zero. By solving the resulting system of equations the following expressions are obtained:
\begin{equation}\label{eq:stimapik}
\pi^*_k=z_{\cdot k}/I, \ \ k=1, \ldots, K,
\end{equation}
and, provided that the matrix $\sum_{i=1}^I \sum_{k=1}^K z_{ik}\mathbf{X}_{ki}\boldsymbol{\Sigma}_{k}^{-1}\mathbf{X}'_{ki}$ is non-singular, \begin{eqnarray}\label{eq:stimagamma}
\boldsymbol{\gamma}^* & = & \left(\sum_{i=1}^I \sum_{k=1}^K z_{ik}\mathbf{X}_{ki}\boldsymbol{\Sigma}_{k}^{-1}\mathbf{X}'_{ki}\right)^{-1}\sum_{i=1}^I \sum_{k=1}^K  z_{ik}\mathbf{X}_{ki}\boldsymbol{\Sigma}_{k}^{-1}\mathbf{y}_i \\ \label{eq:stimasigmak}
\boldsymbol{\Sigma}_k^* & = & z_{\cdot k}^{-1}\mathbf{S}_k, \ \ k=1, \ldots, K.
\end{eqnarray}

Using some initial value for $\boldsymbol{\theta}$, say $\boldsymbol{\theta}^{(0)}$, the E-step on the $(r+1)$th iteration of the EM algorithm is effected by simply replacing $z_{ik}$ by $E_{\boldsymbol{\theta}^{(r)}}(z_{ik}|\mathbf{y}_i, \mathbf{x}_i)= Pr_{\boldsymbol{\theta}^{(r)}}(z_{ik}=1|\mathbf{y}_i, \mathbf{x}_i)=p_{ik}^{(r)}$, which is the posterior probability that $\mathbf{y}_i$ is generated from the $k$th component of the mixture. Namely:
\begin{equation} \label{stimazik} \nonumber
p_{ik}^{(r)}=\frac{\pi^{(r)}_k \phi_D\left(\boldsymbol{y}_i;\boldsymbol{\lambda}^{(r)}_{k}+\mathbf{X}'_{i}\boldsymbol{\beta}^{(r)},\boldsymbol{\Sigma}^{(r)}_{k}\right)}
{\sum_{h=1}^K \pi^{(r)}_h \phi_D\left(\boldsymbol{y}_i;\boldsymbol{\lambda}^{(r)}_{h}+\mathbf{X}'_{i}\boldsymbol{\beta}^{(r)},\boldsymbol{\Sigma}^{(r)}_{h}\right)}.
\end{equation}
On the M-step at the $(r+1)$th iteration of the EM algorithm, the updated estimates of the model parameters $\pi^{(r+1)}_k$, $\boldsymbol{\gamma}^{(r+1)}$ and $\boldsymbol{\Sigma}^{(r+1)}_{k}$ are computed using equations (\ref{eq:stimapik}), (\ref{eq:stimagamma}) and (\ref{eq:stimasigmak}), respectively, where $z_{ik}$ is replaced by $p_{ik}^{(r)}$. As equation (\ref{eq:stimagamma}) depends on the $\boldsymbol{\Sigma}_{k}$'s and equation (\ref{eq:stimasigmak}) depends on $\boldsymbol{\gamma}$, the updated estimates of such parameters at the $(r+1)$th iteration are obtained through an iterative process in which the estimate of $\boldsymbol{\gamma}$ is updated, given an estimate of the $\boldsymbol{\Sigma}_{k}$'s, and vice versa, until convergence.
As far as the choice of $\boldsymbol{\theta}^{(0)}$ is concerned, several strategies can be used \citep[see, e.g.,][]{galimberti2013}. For example, multiple random initializations can be considered. Otherwise, $\boldsymbol{\beta}^{(0)}$ can be obtained by fitting the standard seemingly unrelated regression model; the sample residuals of this model can be used to derive starting values for the remaining parameters, for example by using them to fit a Gaussian mixture model.

\section{Experimental results}\label{sec:experimental}
The usefulness and effectiveness of the methods described in Section~\ref{sec:methods} are illustrated through the analysis of the Australian Institute of Sport (AIS) dataset \citep{cook1994}. Namely, the interest is focused on studying the joint linear dependence of four biometrical variables (body mass index (BMI), sum of skin folds (SSF), percentage of body fat (PBF), lean body mass (LBM)) on three variables providing information about blood composition (red cell count (RCC), white cell count (WCC), plasma ferritine concentration (PFC)). The same problem was investigated by \citet{soffritti2011} using multivariate linear regression models.

A first study is performed to select the regressors to be used for each biometrical variable in a seemingly unrelated linear regression model given by equation (\ref{eq:mixreg}). The main results are summarized in Section~\ref{sec:ricercaregressori}.
Properties of the ML estimates of the regression coefficients for the selected model are numerically evaluated (see Section~\ref{sec:valutaznumerica}). All analyses are performed in the \verb"R" environment \citep{R2012}. A specific function implementing the ML estimation through the EM algorithm and the calculation of the Hessian matrix is used. The starting values of the model parameters are obtained through a strategy that fits Gaussian mixture models to the sample residuals of the classical seemingly unrelated linear regression model. The EM algorithm is stopped when the number of iterations reaches 500 or $|l_\infty ^{(r+1)}-l^{(r)}|<10^{-8}$, where $l^{(r)}$ is the log-likelihood value from iteration $r$, and $l_\infty^{(r+1)}$ is the asymptotic estimate of the log-likelihood at iteration $r+1$ \citep{mcnicholas2008}. The stopping rules for each M step are either when the mean Euclidean distance between two consecutive estimated vectors of the model parameters is lower than $10^{-8}$ or when the number of iterations reaches the maximum of 500.

\subsection{Selection of the regressors}\label{sec:ricercaregressori}
Seemingly unrelated linear regression models from equation (\ref{eq:mixreg}) are
estimated for $K=1, 2, 3$. For each of these values, an exhaustive search is performed to select the relevant regressors for each of the $D=4$ dependent variables. Namely, for each value of $K$, $2^{3 \cdot D}=4096$ different regression models are fitted to the dataset, thus resulting in $12288$ different seemingly unrelated linear regression models. The total number $P$ of regressors included in a model is between 0 and 12.

The EM algorithm has failed due to the singularity of some matrices for two models when $K=2$ and 40 models when $K=3$. The choice of the best model among the estimated ones is performed using the Bayesian Information
Criterion \citep{schwarz1978}:
\begin{equation}\nonumber
BIC_M=2 \max \left[l_M \right]-\mathrm{npar}_M\log(I),
\end{equation}
where $\max \left[l_M \right]$ is the
maximum of the log-likelihood of a model $M$ for the given sample of
$I$ observations, and $\mathrm{npar}_M$ is the number of unconstrained parameters to be
estimated for that model. This criterion allows to trade-off the
fit and parsimony of a given model: the greater the $BIC$, the
better the model.

\begin{figure}
 \centering
   \includegraphics[width=0.6\textwidth]{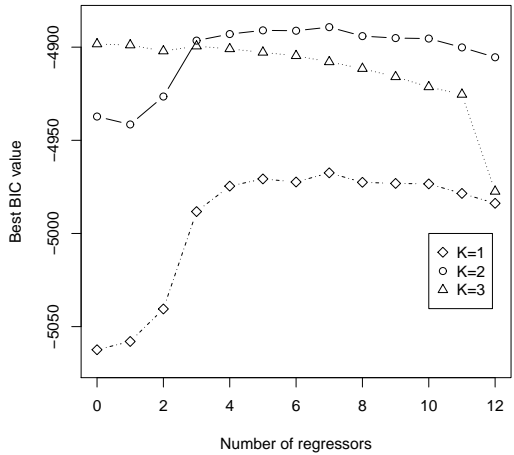}
  \caption{Best BIC values by total number of regressors and number of components.}\label{fig:BIC}
\end{figure}

Figure~\ref{fig:BIC} shows the $BIC$ values of the fitted models with the best trade-off (i.e., the maximum value of the $BIC$) among all the models having the same values of $K$ and $P$, for $K=1, 2, 3$ and $P=0, \ldots, 12$.
By comparing models having the same value of $P$ it emerges that the best performance is obtained using models with three components when the total number of regressors is low ($P=0, 1, 2$); otherwise, models with two components should be preferred. Thus, the introduction of a finite mixture for the distribution of the error terms allows to obtain a relevant improvement with respect to the classical seemingly unrelated regression model with Gaussian errors, for all $P$.

If models are compared by controlling the number of components, $P=7$ regressors should be used when $K=1,2$. Namely, for both values of $K$, the selected regressors for the equations of the variables BMI, PBF and LBM are RCC and PFC; only RCC is selected as a relevant regressor for the equation of SSF. When $K=3$, the best trade-off is obtained using a model without regressors. Some results concerning these three latter models are illustrated in Table~\ref{t:ais1}. Overall, according to the $BIC$ the best model is the one with $K =
2$. In this model, the estimates of the parameters $\pi_1$ and $\pi_2$ are $0.619$ and $0.381$. Tables~\ref{t:ais2} and \ref{t:ais3} report the estimates of the remaining parameters. Compared to the second component, the first component is characterized by lower values of the intercepts for all dependent variables and lower variances for BMI, SSF and PBF. Further differences between components concern some correlations (see the lower triangular parts of $\hat{\boldsymbol{\Sigma}}_1$ and $\hat{\boldsymbol{\Sigma}}_2$ in Table~\ref{t:ais2}). The estimated standard errors of the ML estimators of the regression coefficients in Table~\ref{t:ais3} are computed as the square root of the diagonal elements of $H(\hat{\boldsymbol{\theta}})^{-1}$ that refer to $\boldsymbol{\beta}$. The asymptotic confidence intervals for the regression coefficients in Table~\ref{t:ais3} are obtained by assuming an asymptotic normal distribution for the ML estimators. None of such intervals contains the 0 value.

\begin{table*}
\caption{Maximized log-likelihood and $BIC$ value for the best models with $K$ components ($K=1,2,3$) fitted to the AIS dataset.}
\label{t:ais1}
\centering
\begin{tabular}{lcccc}
\hline\noalign{\smallskip}
$K$ & $P$ & $l_M(\hat{\boldsymbol{\theta}})$ & $\mathrm{npar}_M$ & $BIC_M$ \\ \noalign{\smallskip}\hline\noalign{\smallskip}
$1$ & 7   & $-$2427.993    &     21        & $-$4967.46\\
$2$ & 7   & $-$2349.083    &     36        & $-$4889.26\\
$3$ & 0   & $-$2332.382    &     44        & $-$4898.33\\
\noalign{\smallskip}\hline
\end{tabular}
\end{table*}

\begin{table*}
\caption{Estimates of parameters $\boldsymbol{\lambda}_k$ and $\boldsymbol{\Sigma}_k$ obtained from the best model fitted to the AIS dataset. Estimated correlation coefficients between dependent variables (in italics) are reported in the lower triangular parts of $\hat{\boldsymbol{\Sigma}}_1$ and $\hat{\boldsymbol{\Sigma}}_2$.}
\label{t:ais2}
\centering
\begin{tabular}{crrrr}
\hline\noalign{\smallskip}
 & BMI & SSF & PBF & LBM \\ \noalign{\smallskip}\hline\noalign{\smallskip}
$\hat{\boldsymbol{\lambda}}_1'$ & 10.04 & 86.57  & 23.19 & $-$7.02 \\
$\hat{\boldsymbol{\lambda}}_2'$ & 12.99 & 136.43 & 32.52 & $-$4.88 \\ \noalign{\smallskip}\hline\noalign{\smallskip}
$\hat{\boldsymbol{\Sigma}}_1$ & 3.96 & 5.14 & $-$0.09 & 18.99 \\
& \textit{0.198} & 169.94 & 31.21 & 2.63 \\
& \textit{$-$0.017} & \textit{0.899} & 7.10 & $-$8.73 \\
& \textit{0.810} & \textit{0.017} & \textit{$-$0.278} & 138.82 \\\noalign{\smallskip}\hline\noalign{\smallskip}
$\hat{\boldsymbol{\Sigma}}_2$ & 6.85 & 17.43 & 0.89 & 14.59 \\
& \textit{0.244} & 744.38 & 107.03 & $-$54.50 \\
& \textit{0.080} & \textit{0.928} & 17.88 & $-$15.05 \\
& \textit{0.681} & \textit{$-$0.244} & \textit{$-$0.435} & 67.07 \\
\noalign{\smallskip}\hline
\end{tabular}
\end{table*}

\begin{table*}
\caption{Estimates of the regression coefficients (r.c.) calculated from the best model fitted to the AIS dataset and their estimated standard errors (s.e.). The asymptotic confidence intervals (c.i.) are computed at the 95\% level of confidence.}
\label{t:ais3}
\centering
\begin{tabular}{lccc}
\hline\noalign{\smallskip}
Dependent variable  &    &  Regressors     & \\
                    &    &  RCC            & PFC \\
\hline\noalign{\smallskip}
BMI                 & r.c. & 2.286          & 0.013\\
                    & s.e. & 0.339          & 0.003\\
                    & c.i. & (1.621, 2.950) & (0.007, 0.019)\\ \hline
SSF                 & r.c. & $-$7.746   & - \\
                    & s.e. & 2.783      & - \\
                    & c.i. & ($-$13.200, $-$2.292) & - \\  \hline
PBF                 & r.c. & $-$2.724            & $-$0.005\\
                    & s.e. & 0.565               & 0.002\\
                    & c.i. & ($-$3.832, $-$1.616) & ($-$0.009, $-$0.001)\\ \hline
LBM                 & r.c. & 14.211           & 0.052\\
                    & s.e. & 1.649            & 0.015\\
                    & c.i. & (10.979, 17.442) & (0.023, 0.082)\\
\noalign{\smallskip}\hline
\end{tabular}
\end{table*}

The best model can be used to assign each athlete to the component of the mixture that registered the highest posterior probability, thus producing a partition of the sample into two clusters. Most of the athletes assigned to the second cluster are female (79.2\%), while 68.8\% of the athletes classified in the first cluster are male (Tab.~\ref{t:ais4}). This classification of the athletes is statistically associated with athletes' gender ($\chi^2=43.96$, $p-value=3.36\cdot10^{-11}$). Thus, the omitted regressor captured by the selected model has an effect which is strongly connected with athletes' gender.

\begin{table*}
\caption{Joint classification of the athletes according to gender and cluster membership estimated by the best model.}
\label{t:ais4}
\centering
\begin{tabular}{lrrr}\hline
          & Gender &      &  \\ \hline
Cluster   & Female & Male &\\ \hline
1       & 39 & 86 & 125\\
2       & 61 & 16 & 77\\
        & 100 & 102 & 202\\ \hline
\end{tabular}
\end{table*}

\subsection{A numerical study of some properties of the ML estimates of $\boldsymbol{\beta}$}\label{sec:valutaznumerica}

Properties of the ML estimates of the regression coefficients are evaluated using the parametric bootstrapping residual method \citep{efron1993}. Namely, 5000 bootstrap samples of $I=202$ observations each are generated from the best seemingly unrelated linear regression model described in Section~\ref{sec:ricercaregressori} with parameters equal to the estimates provided in Tables~\ref{t:ais2} and \ref{t:ais3}. For each sample, the ML estimates of the model parameters are computed. For two bootstrap samples this computation is not performed due to the singularity of some matrices. Table~\ref{t:ais4} provides the means and standard deviations of the 4998 ML estimates of the regression coefficients as well as the bootstrap 95\% confidence intervals obtained using the percentile method.

\begin{table*}
\caption{Means and standard deviations (s.d.) of bootstrap replicates of the regression coefficients for the best model fitted to the AIS dataset. Bootstrap confidence intervals are computed at the 95\% level of confidence.}
\label{t:ais4}
\centering
\begin{tabular}{lccc}
\hline\noalign{\smallskip}
Dependent variable &    &  Regressors     & \\
                   &    & RCC            & PFC \\
\hline\noalign{\smallskip}
BMI                & means & 2.287          & 0.013\\
                   & s.d.  & 0.327          & 0.003\\
                   & c.i.  & (1.656, 2.932) & (0.007, 0.019)\\ \hline
SSF                & means & $-$7.746   & - \\
                   & s.d.  & 2.661      & -\\
                   & c.i.  & ($-$13.066, $-$2.529) & -\\  \hline
PBF                & means & $-$2.724            & $-$0.005\\
                   & s.d. & 0.533               & 0.002\\
                   & c.i. & ($-$3.795, $-$1.697) & ($-$0.009, $-$0.001)\\ \hline
LBM                & means & 14.198           & 0.052\\
                   & s.d. & 1.526            & 0.014\\
                   & c.i. & (11.222, 17.091) & (0.024, 0.080)\\
\noalign{\smallskip}\hline
\end{tabular}
\end{table*}

The comparison between the results in Tables~\ref{t:ais3} and \ref{t:ais4} allows to obtain a numerical evaluation of the properties of the ML estimator for the parameters of the selected model. From the differences between the estimated regression coefficients and the means of the bootstrap replicates it emerges that the bias of the ML estimator is negligible for all regression coefficients. Namely, all the ratios between the absolute value of each bias and the bootstrap estimate for the corresponding standard error are lower than 0.04. As far as the estimates of the standard errors are concerned, the relative differences between the asymptotic and bootstrap estimates range from $-3.4 \%$ (regression coefficient of PFC on PBF: 0.001863 against 0.001929) to $8.1\%$ (regression coefficient of RCC on LBM). These differences in the estimates of the standard errors reflect upon the differences in the confidence intervals: the asymptotic confidence intervals are narrower than the bootstrap intervals when the asymptotic standard errors are smaller than the corresponding bootstrap ones. It is worth noting that the ML estimates are almost in the centre of the corresponding bootstrap confidence intervals. These results are related to the shape of the p.d.f. of the ML estimators. Figures~\ref{fig:rcc} and \ref{fig:fe} show the estimates of these p.d.f.'s obtained by applying the kernel method to the bootstrap replicates (the bandwidths were selected according to \citet{sheather1991}); ML estimates are depicted using vertical dotted lines. The distributions result to be approximately symmetric about the ML estimates.

\begin{figure}
 \centering
   \includegraphics[width=0.8\textwidth]{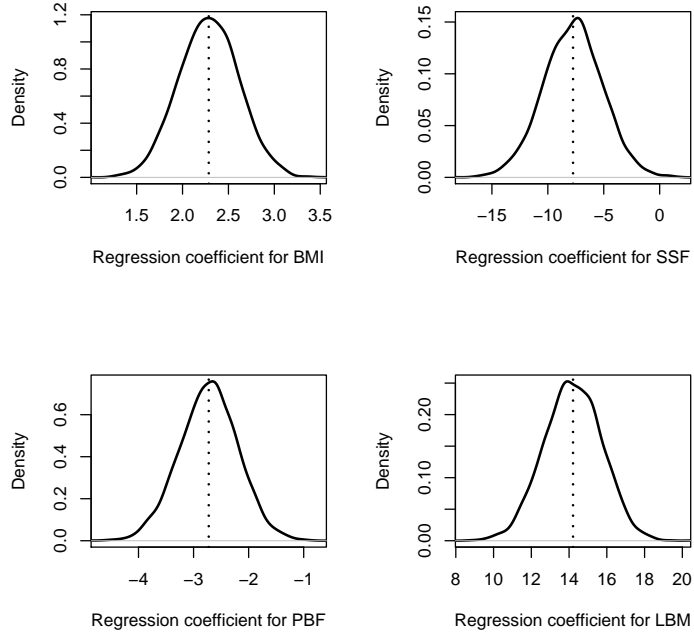}
  \caption{Estimated p.d.f. of the ML estimators of the regression coefficients of RCC on BMI, SSF, PBF and LBM, based on the bootstrap samples.}\label{fig:rcc}
\end{figure}

\begin{figure}
 \centering
   \includegraphics[width=1\textwidth]{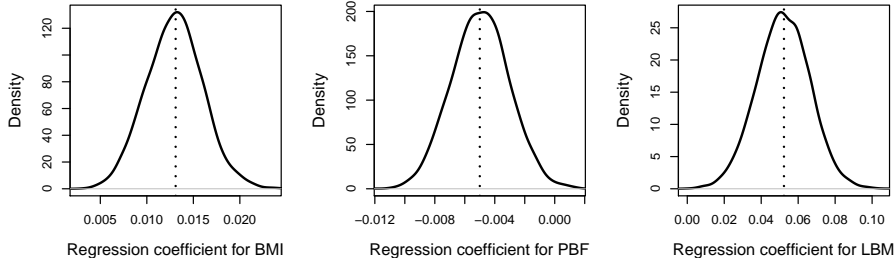}
  \caption{Estimated p.d.f. of the ML estimators of the regression coefficients of PFC on BMI, PBF and LBM, based on the bootstrap samples.}\label{fig:fe}
\end{figure}

\section{Concluding remarks}\label{sec:conclusion}

In this paper, multivariate Gaussian mixtures are used to model the error terms in seemingly unrelated linear regressions. This allows to exploit the flexibility of mixtures for dealing with non Gaussian errors. In particular,
the resulting models are able to handle asymmetric and heavy-tailed errors and to detect and capture the effect of relevant nominal regressors omitted from the model. Furthermore, by setting the number of components equal to one or by constraining all the equations to have the same regressors, some solutions already described in the statistical literature can be obtained as special cases.

Parsimonious seemingly unrelated linear regression models can be obtained by introducing some constraints on the component covariance matrices $\boldsymbol{\Sigma}_{k}$'s, based on the spectral decomposition \citep[see, e.g.,][]{banfield1993,celeux1995,mclachlan2003,mcnicholas2008}. Such models could provide a good fit for some datasets by using a lower number of parameters; they could be useful especially in the presence of a large number of dependent variables.

In Section~\ref{sec:experimental} the $BIC$ is used to select the relevant regressors in each equation as well as the number of mixture components. The use of this criterion can be motivated on the basis of both theoretical and practical results \citep[see, e.g.,][]{cutler1994,keribin2000,ray2008,maugis2009a,maugis2009b}. Clearly, other model selection criteria could be used, such as the $ICL$ (\citet{biernacki2000}, which additionally takes into account the uncertainty of the classification of the sample units to the mixture components.


Some computational issues could arise when using the models proposed in this paper. For example, when the number of candidate regressors is large, an exhaustive search for the relevant regressors for each equation could be unfeasible. A possible solution could be obtained by resorting to stochastic search techniques, such as genetic algorithms \citep[see, e.g.,][]{chatterjee1996}. As far as the EM algorithm is concerned, different initialisation strategies may be considered and evaluated \citep[see, e.g.,][]{biernacki2003,melnykov2012}. Although these issues were not the main focus of this paper, they could deserve further investigation.

\appendix
\section{Proof of Theorem~\ref{theo:score}}\label{proof:score}

The proof is based on the computation of the first order differential of $l\left(\boldsymbol{\theta}\right)$.
The model log-likelihood in equation (\ref{eq:loglik}) can be expressed as $l\left(\boldsymbol{\theta}\right)=\sum_{i=1}^I \ln \left(\sum_{k=1}^K f_{ki}\right)$. Thus, the first differential of $l\left(\boldsymbol{\theta}\right)$ is
\begin{equation}\label{diff}
\mathrm{d}l\left(\boldsymbol{\theta}\right) = \sum_{i=1}^I \mathrm{d}\ln \left(\sum_{k=1}^K f_{ki}\right) = \sum_{i=1}^I \left(\sum_{k=1}^K \alpha_{ki}\mathrm{d}\ln f_{ki}\right).
\end{equation}
Up to an additive constant, $\ln f_{ki}$ is equal to
\begin{equation}\nonumber
\ln\pi_k -\frac{1}{2}\ln \det\left( \boldsymbol{\Sigma}_k\right) -\frac{1}{2}\mathrm{tr}\left[\boldsymbol{\Sigma}_k^{-1}\left(\mathbf{y}_i-\boldsymbol{\lambda}_k-\mathbf{X}'_{i}\boldsymbol{\beta} \right)\left(\mathbf{y}_i-\boldsymbol{\lambda}_k-\mathbf{X}_{i}'\boldsymbol{\beta} \right)'\right],
\end{equation}
and
\begin{equation}\label{eq:dlnfki}
\mathrm{d}\ln f_{ki} = \mathrm{d} \ln \pi_k+\mathrm{d}_{ki1}+\mathrm{d}_{ki2}+\mathrm{d}_{ki3},
\end{equation}
where
\begin{eqnarray}
\mathrm{d}_{ki1} & = & -\frac{1}{2}\mathrm{d}\left(\ln \det\left( \boldsymbol{\Sigma}_k\right)\right),\\
\mathrm{d}_{ki2} & = & -\frac{1}{2} \mathrm{tr}\left[\mathrm{d}\left(\boldsymbol{\Sigma}_k^{-1}\right)\left(\mathbf{y}_i-\boldsymbol{\lambda}_k-\mathbf{X}'_{i}\boldsymbol{\beta} \right)\left(\mathbf{y}_i-\boldsymbol{\lambda}_k-\mathbf{X}'_{i}\boldsymbol{\beta} \right)'\right],\\
\mathrm{d}_{ki3} & = & -\frac{1}{2} \mathrm{tr}\left[\boldsymbol{\Sigma}_k^{-1}\mathrm{d}\left(\left(\mathbf{y}_i-\boldsymbol{\lambda}_k-\mathbf{X}'_{i}\boldsymbol{\beta} \right)\left(\mathbf{y}_i-\boldsymbol{\lambda}_k-\mathbf{X}'_{i}\boldsymbol{\beta} \right)'\right)\right].
\end{eqnarray}
The four terms in equation (\ref{eq:dlnfki}) can be re-expressed as follows:
\begin{eqnarray}\label{eq:dpik}
\mathrm{d}\ln\pi_k & = & \left(\mathrm{d}\boldsymbol{\pi}\right)'\mathbf{a}_k,\\ \label{eq:dk1}
\mathrm{d}_{ki1} & = & -\frac{1}{2}\mathrm{tr}\left[\left(\mathrm{d}\boldsymbol{\Sigma}_k\right)\boldsymbol{\Sigma}_k^{-1}\right],\\\label{eq:dk2}
\mathrm{d}_{ki2} & = & \frac{1}{2} \mathrm{tr}\left[\left(\mathrm{d}\boldsymbol{\Sigma}_k\right)\mathbf{b}_{ki}\mathbf{b}'_{ki}
\right],\\\label{eq:dk3}
\mathrm{d}_{ki3} & = & \left(\mathrm{d}\boldsymbol{\lambda}_k\right)'\mathbf{b}_{ki}
+ \left(\mathrm{d}\boldsymbol{\beta}\right)'\mathbf{X}_{i}\mathbf{b}_{ki},
\end{eqnarray}
where equations~(\ref{eq:dk1})-(\ref{eq:dk3}) are obtained by exploiting some results from matrix derivatives (\citealt[][pgs. 182-183]{magnus1988}; \citealt[][pgs. 292, 293, 361]{schott2005}).
Since the sum of $\mathrm{d}_{ki1}$ and $\mathrm{d}_{ki2}$ results in
\begin{equation}\label{eq:dk12}
\mathrm{d}_{ki1}+\mathrm{d}_{ki2}  =  -\frac{1}{2}\mathrm{d}\left(\mathrm{v}\boldsymbol{\Sigma}_k\right)'\mathbf{G}'\mathrm{vec}\left(\mathbf{B}_{ki}\right),
\end{equation}
\citep[see][pgs. 293, 313, 356, 374]{schott2005}, inserting equations (\ref{eq:dpik}), (\ref{eq:dk3}) and (\ref{eq:dk12}) in equation (\ref{eq:dlnfki}) leads to
\begin{eqnarray}\nonumber
\mathrm{d}\ln f_{ki} & = & \left(\mathrm{d}\boldsymbol{\pi}\right)'\mathbf{a}_k+
\left(\mathrm{d}\boldsymbol{\beta}\right)'\mathbf{X}_{i}\mathbf{b}_{ki}
+ \left(\mathrm{d}\boldsymbol{\lambda}_k\right)'\mathbf{b}_{ki}
-\frac{1}{2}\mathrm{d}\left(\mathrm{v}\boldsymbol{\Sigma}_k\right)'\mathbf{G}'\mathrm{vec}\left(\mathbf{B}_{ki}\right)\\ \label{eq:dlnfki2}
& = & \left(\mathrm{d}\boldsymbol{\pi}\right)'\mathbf{a}_k+
\left(\mathrm{d}\boldsymbol{\beta}\right)'\mathbf{X}_{i}\mathbf{b}_{ki}
+ \left(\mathrm{d}\boldsymbol{\theta}_k\right)'\mathbf{c}_{ki}.
\end{eqnarray}
Using equations (\ref{diff}) and (\ref{eq:dlnfki2}), $\mathrm{d}l\left(\boldsymbol{\theta}\right)$  can be expressed as \begin{equation}
\mathrm{d}l\left(\boldsymbol{\theta}\right)  = \left(\mathrm{d}\boldsymbol{\pi}\right)'\sum_{i=1}^I \sum_{k=1}^K \alpha_{ki}\mathbf{a}_k
  + \left(\mathrm{d}\boldsymbol{\beta}\right)'\sum_{i=1}^I \mathbf{X}_{i}\sum_{k=1}^K \alpha_{ki}\mathbf{b}_{ki}
  +\sum_{k=1}^K \left(\mathrm{d}\boldsymbol{\theta}_k\right)' \sum_{i=1}^I \alpha_{ki}\mathbf{c}_{ki},
\end{equation}
thus proving the theorem.

\section{Proof of Theorem~\ref{theo:hessian}}\label{proof:hessian}
The proof is based on the computation of the second order differential of $l\left(\boldsymbol{\theta}\right)$:
\begin{equation}\label{eq:d1}
\mathrm{d}^2l\left(\boldsymbol{\theta}\right)  =  \sum_{i=1}^I \mathrm{d}^2\ln \left(\sum_{k=1}^K f_{ki}\right),
\end{equation}
where
\begin{equation}\label{eq:d2}
\mathrm{d}^2\ln \left(\sum_{k=1}^K f_{ki}\right) = \sum_{k=1}^K \alpha_{ki}\mathrm{d}^2\ln f_{ki}+\sum_{k=1}^K \alpha_{ki}\left(\mathrm{d}\ln f_{ki}\right)^2  -\left(\sum_{k=1}^K \alpha_{ki}\mathrm{d}\ln f_{ki}\right)^2
\end{equation}
\citep[see][Appendix]{boldea2009}.

Since $\left(\mathrm{d}\ln f_{ki}\right)^2=\left(\mathrm{d}\ln f_{ki}\right)\left(\mathrm{d}\ln f_{ki}\right)'$, using equation (\ref{eq:dlnfki2}) it results that
\begin{eqnarray}\nonumber
\left(\mathrm{d}\ln f_{ki}\right)^2  & = &
\left(\mathrm{d}\boldsymbol{\pi}\right)'\mathbf{a}_k\mathbf{a}'_k\mathrm{d}\boldsymbol{\pi} +  \left(\mathrm{d}\boldsymbol{\pi}\right)'\mathbf{a}_k \mathbf{b}'_{ki}\mathbf{X}'_{i}\mathrm{d}\boldsymbol{\beta} +  \left(\mathrm{d}\boldsymbol{\pi}\right)'\mathbf{a}_k \mathbf{c}'_{ki}\mathrm{d}\boldsymbol{\theta}_k  \\ \nonumber
 & & + \left(\mathrm{d}\boldsymbol{\beta}\right)'\mathbf{X}_{i}\mathbf{b}_{ki}\mathbf{a}'_k\mathrm{d}\boldsymbol{\pi} + \left(\mathrm{d}\boldsymbol{\beta}\right)'\mathbf{X}_{i}\mathbf{b}_{ki}\mathbf{b}'_{ki}\mathbf{X}'_{i}\mathrm{d}\boldsymbol{\beta}+ \left(\mathrm{d}\boldsymbol{\beta}\right)'\mathbf{X}_{i}\mathbf{b}_{ki}\mathbf{c}'_{ki}\mathrm{d}\boldsymbol{\theta}_k \\ \label{eq:perd2b}
 & & + \left(\mathrm{d}\boldsymbol{\theta}_k\right)'\mathbf{c}_{ki}\mathbf{a}'_k\mathrm{d}\boldsymbol{\pi} +  \left(\mathrm{d}\boldsymbol{\theta}_k\right)'\mathbf{c}_{ki}\mathbf{b}'_{ki}\mathbf{X}'_{i}\mathrm{d}\boldsymbol{\beta}  + \left(\mathrm{d}\boldsymbol{\theta}_k\right)'\mathbf{c}_{ki}\mathbf{c}'_{ki}\mathrm{d}\boldsymbol{\theta}_k.
\end{eqnarray}
Similarly,
\begin{eqnarray}\nonumber
\left(\sum_{k=1}^K \alpha_{ki} \mathrm{d} \ln f_{ki}\right)^2 & = & \left( \sum_{k=1}^K \alpha_{ki}\mathrm{d} \ln f_{ki}\right) \left(\sum_{k=1}^K\alpha_{ki}\mathrm{d}\ln f_{ki}\right)'\\ \nonumber
 & = & \left(\mathrm{d}\boldsymbol{\pi}\right)'\bar{\mathbf{a}}_i\bar{\mathbf{a}}'_i\mathrm{d}\boldsymbol{\pi} +
 \left(\mathrm{d}\boldsymbol{\pi}\right)'\bar{\mathbf{a}}_i \bar{\mathbf{b}}'_{i}\mathbf{X}'_{i}\mathrm{d}\boldsymbol{\beta} +  \left(\mathrm{d}\boldsymbol{\pi}\right)'\bar{\mathbf{a}}_i \sum_{k=1}^K \alpha_{ki}\mathbf{c}'_{ki}\mathrm{d}\boldsymbol{\theta}_k \\ \nonumber
   & & + \left(\mathrm{d}\boldsymbol{\beta}\right)'\mathbf{X}_{i}\bar{\mathbf{b}}_{i}\bar{\mathbf{a}}'_i\mathrm{d}\boldsymbol{\pi}+ \left(\mathrm{d}\boldsymbol{\beta}\right)'\mathbf{X}_{i}\bar{\mathbf{b}}_{i}\bar{\mathbf{b}}'_{i}\mathbf{X}'_{i}\mathrm{d}\boldsymbol{\beta}  \\ \nonumber
   & &    + \left(\mathrm{d}\boldsymbol{\beta}\right)'\mathbf{X}_{i}\bar{\mathbf{b}}_{i}\sum_{k=1}^K\alpha_{ki}\mathbf{c}'_{ki}\mathrm{d}\boldsymbol{\theta}_k+ \left[\sum_{k=1}^K\left(\mathrm{d}\boldsymbol{\theta}_k\right)'\alpha_{ki}\mathbf{c}_{ki}\right]\bar{\mathbf{a}}'_i\mathrm{d}\boldsymbol{\pi} \\ \nonumber
    & &    +  \left[\sum_{k=1}^K\left(\mathrm{d}\boldsymbol{\theta}_k\right)'\alpha_{ki}\mathbf{c}_{ki}\right]\bar{\mathbf{b}}'_{i}
  \mathbf{X}_{i}'\mathrm{d}\boldsymbol{\beta}  \\ \label{eq:perd2a}
& &   + \sum_{k=1}^K\sum_{h=1}^K\left(\mathrm{d}\boldsymbol{\theta}_k\right)'\alpha_{ki}\alpha_{hi}\mathbf{c}_{ki}\mathbf{c}'_{hi}
  \mathrm{d}\boldsymbol{\theta}_l.
 \end{eqnarray}
Furthermore,
\begin{eqnarray}\nonumber
\mathrm{d}^2\ln f_{ki} & = & -\left(\mathrm{d}\boldsymbol{\pi}\right)'\mathbf{a}_k\mathbf{a}'_k\mathrm{d}\boldsymbol{\pi}
-\left(\mathrm{d}\boldsymbol{\beta}\right)'\mathbf{X}_{i}\boldsymbol{\Sigma}_k^{-1}\mathbf{X}'_{i}\mathrm{d}\boldsymbol{\beta}
-\left(\mathrm{d}\boldsymbol{\theta}_k\right)'\mathbf{F}'_{ki}\mathbf{X}'_{i}\mathrm{d}\boldsymbol{\beta}\\ \label{eq:d2lnfki}
& & -\left(\mathrm{d}\boldsymbol{\beta}\right)'\mathbf{X}_{i}\mathbf{F}_{ki}\mathrm{d}\boldsymbol{\theta}_k  -\left(\mathrm{d}\boldsymbol{\theta}_k\right)'\mathbf{C}_{ki}\mathrm{d}\boldsymbol{\theta}_k
\end{eqnarray}
(see \ref{sec:d2lnfki}).
From equations (\ref{eq:d2}), (\ref{eq:perd2b}), (\ref{eq:perd2a}) and (\ref{eq:d2lnfki}) and by grouping together the common factors it follows that
\begin{eqnarray}\nonumber
\mathrm{d}^2\ln \left(\sum_{k=1}^K f_{ki}\right) & = & -\left(\mathrm{d}\boldsymbol{\pi}\right)'\bar{\mathbf{a}}_i\bar{\mathbf{a}}'_i\mathrm{d}\boldsymbol{\pi}  +  \left(\mathrm{d}\boldsymbol{\pi}\right)'\left[\left(\sum_{k=1}^K\alpha_{ki}\mathbf{a}_k \mathbf{b}'_{ki}\right)-\bar{\mathbf{a}}_i \bar{\mathbf{b}}'_{i}\right]\mathbf{X}'_{i}\mathrm{d}\boldsymbol{\beta} \\ \nonumber
& & + \left(\mathrm{d}\boldsymbol{\pi}\right)'\left[\sum_{k=1}^K\alpha_{ki}\left(\mathbf{a}_{k}-\bar{\mathbf{a}}_i\right) \mathbf{c}'_{ki}\mathrm{d}\boldsymbol{\theta}_k\right]\\ \nonumber
& & + \left(\mathrm{d}\boldsymbol{\beta}\right)'\mathbf{X}_{i}\left[\left(\sum_{k=1}^K\alpha_{ki}\mathbf{b}_{ki} \mathbf{a}'_{k}\right)-\bar{\mathbf{b}}_{i}\bar{\mathbf{a}}'_i\right]\mathrm{d}\boldsymbol{\pi} \\ \nonumber
  & & -  \left(\mathrm{d}\boldsymbol{\beta}\right)'\mathbf{X}_{i}\left[\bar{\mathbf{B}}_{i}+ \bar{\mathbf{b}}_{i}\bar{\mathbf{b}}'_{i}\right]\mathbf{X}'_{i}\mathrm{d}\boldsymbol{\beta} \\ \nonumber
 & & -\left(\mathrm{d}\boldsymbol{\beta}\right)'\mathbf{X}_{i}\left\lbrace\sum_{k=1}^K\alpha_{ki}
  \left[\mathbf{F}_{ki}-\left(\mathbf{b}_{ki}-\bar{\mathbf{b}}_{i}\right)\mathbf{c}'_{ki}\right]\mathrm{d}\boldsymbol{\theta}_k\right\rbrace \\ \nonumber
  & & + \left[\sum_{k=1}^K\left(\mathrm{d}\boldsymbol{\theta}_k\right)'\alpha_{ki}\mathbf{c}_{ki} \left(\mathbf{a}'_k-\bar{\mathbf{a}}'_i\right)\right]\mathrm{d}\boldsymbol{\pi}\\ \nonumber
 & &   - \left\lbrace\sum_{k=1}^K\left(\mathrm{d}\boldsymbol{\theta}_k\right)'\alpha_{ki}\left[\mathbf{F}'_{ki} -\mathbf{c}_{ki}\left(\mathbf{b}'_{ki}-\bar{\mathbf{b}}'_{i}\right)\right]\right\rbrace\mathbf{X}'_{i}\mathrm{d}\boldsymbol{\beta}\\ \nonumber
    & & - \sum_{k=1}^K\left(\mathrm{d}\boldsymbol{\theta}_k\right)'\alpha_{ki}\left[\mathbf{C}_{ki}-\mathbf{c}_{ki}\mathbf{c}_{ki}'\right]\mathrm{d}\boldsymbol{\theta}_k\\ \label{eq:perd2}
& & -\sum_{k=1}^K\sum_{h=1}^K\left[\left(\mathrm{d}\boldsymbol{\theta}_k\right)'\alpha_{ki}\alpha_{hi}\mathbf{c}_{ki}\mathbf{c}_{hi}'\mathrm{d}\boldsymbol{\theta}_h\right].
 \end{eqnarray}
Inserting equation (\ref{eq:perd2}) in equation (\ref{eq:d1}) completes the proof.

\section{Second order differential of $\ln f_{ki}$}\label{sec:d2lnfki}
Using equation (\ref{eq:dlnfki}) the second order differential of $\ln f_{ki}$ can be expressed as
\begin{equation}\label{eq:d2lnfki}
\mathrm{d}^2\ln f_{ki}=\mathrm{d}^2\ln\pi_k+\mathrm{d}\left(\mathrm{d}_{ki1}\right)+\mathrm{d}\left(\mathrm{d}_{ki2}\right)+\mathrm{d}\left(\mathrm{d}_{ki3}\right).
\end{equation}
From equation (\ref{eq:dpik}) it follows that
\begin{equation}\label{eq:d2pik}
\mathrm{d}^2\ln\pi_k=-\left(\mathrm{d}\boldsymbol{\pi}\right)'\mathbf{a}_k\mathbf{a}'_k\mathrm{d}\boldsymbol{\pi}.
\end{equation}
The second term in equation (\ref{eq:d2lnfki}) is equal to
\begin{equation}\label{eq:ddki1}
\mathrm{d}\left(\mathrm{d}_{ki1}\right)=-\frac{1}{2}\mathrm{tr}\left[\mathrm{d}\boldsymbol{\Sigma}_k\left(\mathrm{d}\boldsymbol{\Sigma}_k^{-1}\right)\right]=\frac{1}{2}\mathrm{tr}\left[\left(\mathrm{d}\boldsymbol{\Sigma}_k\right)\boldsymbol{\Sigma}_k^{-1}\left(\mathrm{d}\boldsymbol{\Sigma}_k\right)\boldsymbol{\Sigma}_k^{-1}\right].
\end{equation}
The third term that composes $\mathrm{d}^2\ln f_{ki}$ results to be
\begin{eqnarray}\nonumber
\mathrm{d}\left(\mathrm{d}_{ki2}\right) & = &
\frac{1}{2}\mathrm{tr}\left[\mathrm{d}\left(\boldsymbol{\Sigma}_k^{-1}\right)\left(\mathrm{d}\boldsymbol{\Sigma}_k\right)\boldsymbol{\Sigma}_k^{-1}\left(\mathbf{y}_i-\boldsymbol{\lambda}_k-\mathbf{X}'_{i}\boldsymbol{\beta} \right)\left(\mathbf{y}_i-\boldsymbol{\lambda}_k-\mathbf{X}'_{i}\boldsymbol{\beta} \right)'\right]\\ \nonumber
& & +\frac{1}{2}\mathrm{tr}\left[\boldsymbol{\Sigma}_k^{-1}\left(\mathrm{d}\boldsymbol{\Sigma}_k\right)\mathrm{d}\left(\boldsymbol{\Sigma}_k^{-1}\right)\left(\mathbf{y}_i-\boldsymbol{\lambda}_k-\mathbf{X}'_{i}\boldsymbol{\beta} \right)\left(\mathbf{y}_i-\boldsymbol{\lambda}_k-\mathbf{X}'_{i}\boldsymbol{\beta} \right)'\right]\\ \nonumber
& &
+\frac{1}{2}\mathrm{tr}\left[\boldsymbol{\Sigma}_k^{-1}\left(\mathrm{d}\boldsymbol{\Sigma}_k\right)
\boldsymbol{\Sigma}_k^{-1}\mathrm{d}\left(\left(\mathbf{y}_i-\boldsymbol{\lambda}_k-\mathbf{X}'_{i}\boldsymbol{\beta} \right)\left(\mathbf{y}_i-\boldsymbol{\lambda}_k-\mathbf{X}'_{i}\boldsymbol{\beta} \right)'\right)\right].
\end{eqnarray}
By exploiting some properties of the trace of a square matrix \citep[see, e.g.][]{schott2005}, $\mathrm{d}\left(\mathrm{d}_{ki2}\right)$ can also be expressed as
\begin{eqnarray}\nonumber
\mathrm{d}\left(\mathrm{d}_{ki2}\right) & = & \mathrm{tr}\left[\left(\mathrm{d}\boldsymbol{\Sigma}_k\right)\mathrm{d}\left(\boldsymbol{\Sigma}_k^{-1}\right)\left(\mathbf{y}_i-\boldsymbol{\lambda}_k-\mathbf{X}'_{i}\boldsymbol{\beta} \right)\left(\mathbf{y}_i-\boldsymbol{\lambda}_k-\mathbf{X}'_{i}\boldsymbol{\beta} \right)^{\top}\boldsymbol{\Sigma}_k^{-1}\right]\\ \nonumber
& & +\frac{1}{2}\mathrm{tr}\left[\boldsymbol{\Sigma}_k^{-1}\left(\mathrm{d}\boldsymbol{\Sigma}_k\right)
\boldsymbol{\Sigma}_k^{-1}\mathrm{d}\left(\left(\mathbf{y}_i-\boldsymbol{\lambda}_k-\mathbf{X}'_{i}\boldsymbol{\beta} \right)\left(\mathbf{y}_i-\boldsymbol{\lambda}_k-\mathbf{X}'_{i}\boldsymbol{\beta} \right)'\right)\right],
\end{eqnarray}
and using two theorems about the vec and trace operators \citep[][Theorems 8.9 and 8.12]{schott2005} it follows that
\begin{eqnarray}\nonumber
\mathrm{d}\left(\mathrm{d}_{ki2}\right) & = &
\mathrm{tr}\left[\left(\mathrm{d}\boldsymbol{\Sigma}_k\right)\mathrm{d}\left(\boldsymbol{\Sigma}_k^{-1}\right)\left(\mathbf{y}_i-\boldsymbol{\lambda}_k-\mathbf{X}'_{i}\boldsymbol{\beta} \right)\left(\mathbf{y}_i-\boldsymbol{\lambda}_k-\mathbf{X}'_{i}\boldsymbol{\beta} \right)'\boldsymbol{\Sigma}_k^{-1}\right]\\ \nonumber
& &
 -\left(\mathrm{d}\boldsymbol{\lambda}_k\right)'\left(\mathbf{b}'_{ki}\otimes \boldsymbol{\Sigma}_k^{-1}\right)\mathrm{d}\left(\mathrm{vec}\boldsymbol{\Sigma}_k\right)\\ \label{eq:ddki2c}
& & -\left(\mathrm{d}\boldsymbol{\beta}\right)'\mathbf{X}_{i}\left(\mathbf{b}'_{ki}\otimes \boldsymbol{\Sigma}_k^{-1}\right)\mathrm{d}\left(\mathrm{vec}\boldsymbol{\Sigma}_k\right).
\end{eqnarray}
From equations (\ref{eq:ddki1}) and (\ref{eq:ddki2c}) it follows that
\begin{eqnarray}\nonumber
\mathrm{d}\left(\mathrm{d}_{ki1}\right)+\mathrm{d}\left(\mathrm{d}_{ki2}\right) & = & \frac{1}{2}\mathrm{tr}\left[\left(\mathrm{d}\boldsymbol{\Sigma}_k\right)\boldsymbol{\Sigma}_k^{-1}\left(\mathrm{d}\boldsymbol{\Sigma}_k\right)\boldsymbol{\Sigma}_k^{-1}\right]\\ \nonumber
& & -\mathrm{tr}\left[\left(\mathrm{d}\boldsymbol{\Sigma}_k\right)\boldsymbol{\Sigma}_k^{-1}\left(\mathrm{d}\boldsymbol{\Sigma}_k\right)
\mathbf{b}_{ki}\mathbf{b}'_{ki}\right]\\ \nonumber
& & -\left(\mathrm{d}\boldsymbol{\lambda}_k\right)'\left(\mathbf{b}'_{ki}\otimes \boldsymbol{\Sigma}_k^{-1}\right)\mathrm{d}\left(\mathrm{vec}\boldsymbol{\Sigma}_k\right)\\ \nonumber
& & -\left(\mathrm{d}\boldsymbol{\beta}\right)'\mathbf{X}_{i}\left(\mathbf{b}'_{ki}\otimes \boldsymbol{\Sigma}_k^{-1}\right)\mathrm{d}\left(\mathrm{vec}\boldsymbol{\Sigma}_k\right)\\ \nonumber
& = & \frac{1}{2}\mathrm{tr}\left\lbrace\left(\mathrm{d}\boldsymbol{\Sigma}_k\right)\boldsymbol{\Sigma}_k^{-1}\left(\mathrm{d}\boldsymbol{\Sigma}_k\right)\left[\boldsymbol{\Sigma}_k^{-1}+\boldsymbol{\Sigma}_k^{-1}-\boldsymbol{\Sigma}_k^{-1}-2\mathbf{b}_{ki}\mathbf{b}'_{ki}\right]\right\rbrace\\ \nonumber
& & -\left(\mathrm{d}\boldsymbol{\lambda}_k\right)'\left(\mathbf{b}'_{ki}\otimes \boldsymbol{\Sigma}_k^{-1}\right)\mathrm{d}\left(\mathrm{vec}\boldsymbol{\Sigma}_k\right)\\ \nonumber
& & -\left(\mathrm{d}\boldsymbol{\beta}\right)'\mathbf{X}_{i}\left(\mathbf{b}'_{ki}\otimes \boldsymbol{\Sigma}_k^{-1}\right)\mathrm{d}\left(\mathrm{vec}\boldsymbol{\Sigma}_k\right)\\ \nonumber
& = & -\frac{1}{2}\mathrm{vec}\left(\left(\mathrm{d}\boldsymbol{\Sigma}_k\right)'\right)'\left[\left(\boldsymbol{\Sigma}_k^{-1}-2\mathbf{B}_{ki}\right)'\otimes \boldsymbol{\Sigma}_k^{-1}\right]\mathrm{vec}\left(\mathrm{d}\boldsymbol{\Sigma}_k\right)\\ \nonumber
& & -\left(\mathrm{d}\boldsymbol{\lambda}_k\right)'\left(\mathbf{b}'_{ki}\otimes \boldsymbol{\Sigma}_k^{-1}\right)\mathrm{d}\left(\mathrm{vec}\boldsymbol{\Sigma}_k\right)\\ \nonumber
& & -\left(\mathrm{d}\boldsymbol{\beta}\right)'\mathbf{X}_{i}\left(\mathbf{b}'_{ki}\otimes \boldsymbol{\Sigma}_k^{-1}\right)\mathrm{d}\left(\mathrm{vec}\boldsymbol{\Sigma}_k\right) \\ \nonumber
& = & -\frac{1}{2}\mathrm{d}\left(\mathrm{v}\boldsymbol{\Sigma}_k\right)'\mathbf{G}'\left[\left(\boldsymbol{\Sigma}_k^{-1}-2\mathbf{B}_{ki}\right)\otimes \boldsymbol{\Sigma}_k^{-1}\right]\mathbf{G}\mathrm{d}\left(\mathrm{v}\boldsymbol{\Sigma}_k\right)\\ \nonumber
& & -\left(\mathrm{d}\boldsymbol{\lambda}_k\right)'\left(\mathbf{b}'_{ki}\otimes \boldsymbol{\Sigma}_k^{-1}\right)\mathbf{G}\mathrm{d}\left(\mathrm{v}\boldsymbol{\Sigma}_k\right)\\ \label{eq:ddki12}
& & -\left(\mathrm{d}\boldsymbol{\beta}\right)'\mathbf{X}_{i}\left(\mathbf{b}'_{ki}\otimes \boldsymbol{\Sigma}_k^{-1}\right)\mathbf{G}\mathrm{d}\left(\mathrm{v}\boldsymbol{\Sigma}_k\right),
\end{eqnarray}
where the third and fourth equalities are obtained using some
properties of the vec operator \citep[see][pg. 294]{schott2005}.

From equation (\ref{eq:dk3}) it is possible to write
\begin{eqnarray}\nonumber
\mathrm{d}\left(\mathrm{d}_{ki3}\right) & = & \left(\mathrm{d}\boldsymbol{\lambda}_k\right)'\mathrm{d}\mathbf{b}_{ki}
+ \left(\mathrm{d}\boldsymbol{\beta}\right)'\mathbf{X}_{i}\mathrm{d}\mathbf{b}_{ki}\\ \nonumber
& = & -\left(\mathrm{d}\boldsymbol{\lambda}_k\right)'\boldsymbol{\Sigma}_k^{-1}\mathrm{d}\left(\boldsymbol{\Sigma}_k\right)\mathbf{b}_{ki}-\left(\mathrm{d}\boldsymbol{\lambda}_k\right)'\boldsymbol{\Sigma}_k^{-1}\mathrm{d}\boldsymbol{\lambda}_k-\left(\mathrm{d}\boldsymbol{\lambda}_k\right)'\boldsymbol{\Sigma}_k^{-1}\mathbf{X}'_{i}\mathrm{d}\boldsymbol{\beta}\\ \nonumber
& & -\left(\mathrm{d}\boldsymbol{\beta}\right)'\mathbf{X}_{i}\boldsymbol{\Sigma}_k^{-1}\mathrm{d}\left(\boldsymbol{\Sigma}_k\right)\mathbf{b}_{ki}-\left(\mathrm{d}\boldsymbol{\beta}\right)'\mathbf{X}_{i}\boldsymbol{\Sigma}_k^{-1}\mathrm{d}\boldsymbol{\lambda}_k-\left(\mathrm{d}\boldsymbol{\beta}\right)^{\top}\mathbf{X}_{i}\boldsymbol{\Sigma}_k^{-1}\mathbf{X}'_{i}\mathrm{d}\boldsymbol{\beta}\\
\nonumber
 & = & -\mathrm{d}\left(\mathrm{v}\boldsymbol{\Sigma}_k\right)'\mathbf{G}'\left(\mathbf{b}_{ki}\otimes \boldsymbol{\Sigma}_k^{-1}\right)\mathrm{d}\boldsymbol{\lambda}_k
  -\left(\mathrm{d}\boldsymbol{\lambda}_k\right)'\boldsymbol{\Sigma}_k^{-1}\mathrm{d}\boldsymbol{\lambda}_k\\ \nonumber
& &  -\left(\mathrm{d}\boldsymbol{\lambda}_k\right)'\boldsymbol{\Sigma}_k^{-1}\mathbf{X}'_{i}\mathrm{d}\boldsymbol{\beta}
  -\mathrm{d}\left(\mathrm{v}\boldsymbol{\Sigma}_k\right)'\mathbf{G}'\left(\mathbf{b}_{ki}\otimes \boldsymbol{\Sigma}_k^{-1}\right)\mathbf{X}'_{i}\mathrm{d}\boldsymbol{\beta}\\ \label{eq:ddki3a}
& & -\left(\mathrm{d}\boldsymbol{\beta}\right)'\mathbf{X}_{i}\boldsymbol{\Sigma}_k^{-1}\mathrm{d}\boldsymbol{\lambda}_k-\left(\mathrm{d}\boldsymbol{\beta}\right)'\mathbf{X}_{i}\boldsymbol{\Sigma}_k^{-1}\mathbf{X}'_{i}\mathrm{d}\boldsymbol{\beta},
\end{eqnarray}
where the third equality results from the same theorems about the vec and trace operators employed above and
the second equality is obtained using the following expression for $\mathrm{d}\mathbf{b}_{ki}$:
\begin{eqnarray}\nonumber
\mathrm{d}\mathbf{b}_{ki} & = & \mathrm{d}\left(\boldsymbol{\Sigma}_k^{-1}\right)\left(\mathbf{y}_i-\boldsymbol{\lambda}_k-\mathbf{X}'_{i}\boldsymbol{\beta} \right)+\boldsymbol{\Sigma}_k^{-1}\mathrm{d}\left(\mathbf{y}_i-\boldsymbol{\lambda}_k-\mathbf{X}'_{i}\boldsymbol{\beta} \right)\\ \nonumber
& = & -\boldsymbol{\Sigma}_k^{-1}\mathrm{d}\left(\boldsymbol{\Sigma}_k\right)\mathbf{b}_{ki}-\boldsymbol{\Sigma}_k^{-1}\mathrm{d}\boldsymbol{\lambda}_k-\boldsymbol{\Sigma}_k^{-1}\mathbf{X}'_{i}\mathrm{d}\boldsymbol{\beta}.
\end{eqnarray}

Inserting equations (\ref{eq:d2pik}), (\ref{eq:ddki12}) and (\ref{eq:ddki3a}) in equation (\ref{eq:d2lnfki})
and using the definitions of $\boldsymbol{\theta}_k$, $\mathbf{F}_{ki}$ and $\mathbf{C}_{ki}$ introduced in Section~\ref{sec:score} results in the following expression for $\mathrm{d}^2\ln f_{ki}$:
\begin{eqnarray}\nonumber
\mathrm{d}^2\ln f_{ki} & = & -\left(\mathrm{d}\boldsymbol{\pi}\right)'\mathbf{a}_k\mathbf{a}'_k\mathrm{d}\boldsymbol{\pi}
-\left(\mathrm{d}\boldsymbol{\beta}\right)'\mathbf{X}_{i}\boldsymbol{\Sigma}_k^{-1}\mathbf{X}'_{i}\mathrm{d}\boldsymbol{\beta}
-\left(\mathrm{d}\boldsymbol{\theta}_k\right)'\mathbf{F}'_{ki}\mathbf{X}'_{i}\mathrm{d}\boldsymbol{\beta}\\ \nonumber
&  & -\left(\mathrm{d}\boldsymbol{\beta}\right)'\mathbf{X}_{i}\mathbf{F}_{ki}\mathrm{d}\boldsymbol{\theta}_k  -\left(\mathrm{d}\boldsymbol{\theta}_k\right)'\mathbf{C}_{ki}\mathrm{d}\boldsymbol{\theta}_k.
\end{eqnarray}

\end{document}